\documentclass[11pt]{article}
\usepackage{epsfig}
\usepackage{amsfonts}
\usepackage{amssymb}
\usepackage{amstext}
\usepackage{amsmath}
\usepackage{xspace}
\usepackage{theorem}

\newcommand{\ignore}[1]{}
\ignore{
\usepackage{pstcol}
\usepackage{pst-text}
\usepackage{pst-node}
\usepackage{epsfig}
}

\newcounter{obs}
\stepcounter{obs}
\ignore{
\definecolor{back-high}{rgb}{1,1,0.6}
\definecolor{back}{rgb}{1,1,.5}
\definecolor{nc}{rgb}{.38,.5,.9}
\definecolor{nc2}{rgb}{.3,.8,.3}
\definecolor{db}{rgb}{0,0,.3}
\definecolor{dr}{rgb}{1,0.5,0}
\definecolor{bg}{rgb}{1,.5,.2}
\definecolor{bg2}{rgb}{1,1,.5}
\definecolor{bgs}{rgb}{.5,1,.5}
\definecolor{mc1}{rgb}{.7,.7,1}
\definecolor{mc2}{rgb}{.8,.8,1}
\definecolor{bgslight}{rgb}{.8,1,.8}
}

\newtheorem{theorem}{Theorem}[section]
\newtheorem{claim}{Claim}
\newtheorem{corollary}{Corollary}
\newtheorem{lemma}[theorem]{Lemma}

\newenvironment{proof}{\noindent{\bf Proof}:}{$\hfill \Box$\\}

\newenvironment{proofof}[1]{\noindent{\bf Proof of #1}:}{$\hfill \Box$\\}

\newcommand{\Z}{\mathbb{Z}}
\newcommand{\R}{\mathbb{R}}
\newcommand{\one}{{\bf{1}}}

\title{Integrality Gap of the Vertex Cover Linear Programming Relaxation}
\author{
Mohit Singh\thanks{Groseclose 410,H. Milton Stewart School of Industrial and Systems Engineering, 755 Ferst Avenue, Atlanta GA.  Email:mohit.singh@isye.gatech.edu.} \\
H. Milton Stewart School of Industrial and Systems Engineering, \\
Georgia Institute of Technology, Atlanta, GA.
}

\begin{document}

\date{}
\maketitle

\begin{abstract}
We give a characterization result for the integrality gap of the
natural linear programming relaxation for the vertex cover
problem. We show that integrality gap of the standard linear programming relaxation for any graph $G$ {equals} $\left(2-\frac{2}{\chi^f(G)}\right)$ where $\chi^f(G)$ denotes the fractional chromatic number of $G$.

\end{abstract}

\section{Introduction}

Given a vertex-weighted graph $G=(V,E)$ with weights
$c:V\rightarrow \mathbb{R}_+$, the vertex cover problem asks for a
minimum weight subset of vertices $S\subseteq V$ such that each edge $e\in
E$ has at least one endpoint in $S$. The vertex cover problem is
NP-hard~\cite{karp} and a 2-approximation algorithm was given by Nemhauser and Trotter~\cite{nt}.

Mathematical programming techniques have been extensively studied
for the vertex cover problem. The following is a natural integer
linear program for the vertex cover problem. We denote by $IP(G)$ the convex hull of all integer feasible vertex covers which is
the convex hull of the feasible solutions to following the integer linear
program.

\begin{eqnarray*}
\label{ip1} \min \hspace{1ex} c^T x &=& \sum\limits_{v\in V} c_v\,x_v \\
\textrm{s.t.} \hspace{5ex} x_u+x_v &\geq& 1 \hspace{14ex} \forall \,
\{u,v\}\in E\\
\label{int1}   x &\in&\Z_+^V
\end{eqnarray*}

Relaxing the integrality constraints we obtain a linear program.
We denote $P(G)$ to be the polyhedron of all feasible solutions to
the following linear programming relaxation to the vertex cover
problem.

\begin{eqnarray*}
\label{lp1} \min \hspace{1ex} c^T x &=&\sum\limits_{v\in V} c_v\,x_v \\
\textrm{ s.t. } \hspace{5ex} x_u+x_v &\geq& 1 \hspace{14ex} \forall \,
\{u,v\}\in E\\
\label{relax}   x_v &\geq& 0  \hspace{15ex} \forall \, v\in V
\end{eqnarray*}

The polyhedron $P(G)$ has been studied extensively. It is known that $P(G)$ is integral if and
only if $G$ is bipartite. This follows as a corollary of the Hungarian method~\cite{kuhn}. Nemhauser and Trotter~\cite{nt} show
that $P(G)$ is half-integral, i.e., for any extreme point $x$
of $P(G)$, we must have $x_v\in \{0,\frac12,1\}$ for each $v\in
V$. This result
directly implies a 2-approximation for the problem by selecting
each vertex with non-zero value in an optimal solution to the linear program.
Surprisingly, this long standing result has not been improved
despite extensive study and is optimal assuming the Unique Games Conjecture~\cite{KhotR}. Moreover, there are no polynomial sized linear programs that approximate the value of the vertex cover better than factor of $2-\epsilon$~\cite{AFPS}.

In this paper, we study the polyhedron $P(G)$ and give a complete
characterization of the integrality gap of the linear program. Our
characterization will also point out which instances are harder
for the vertex cover problem. Surprisingly, we obtain a strong
connection between the integrality gap of the linear program and
another parameter of the graph known as fractional chromatic
number of the graph.

\subsection{Definitions and Preliminaries}

The \emph{integrality gap} of vertex cover LP relaxation is
defined as \[\rho(G)=\max_{c:V\rightarrow \R_+} \frac{\min \{c^Tx:x\in
IP(G)\}}{\min \{c^Tx: x\in P(G)\}}\].

We also use the following result which is folklore and also occurs
explicitly in Goemans~\cite{goemansstrength}. A polyhedron $P\subseteq \R^n_+$ is called blocking type if for any $x\in P$ and $y\geq x$, we have $y\in P$.

\begin{lemma}~\cite{goemansstrength}\label{goemanslem}
Given a blocking type polyhedron $Q$ and its relaxation $P$, the
integrality gap of this relaxation \[\max_{c:V\rightarrow \R_+}
\left\{\frac{\min \left\{c^Tx:x\in Q\right\}}{\min \left\{c^Tx: x\in P\right\}}\right\} =
 \min\left\{\rho\geq 0: \rho \cdot x\in Q| \forall x\in P\right\}.\]
   \end{lemma}
Observe that, since $P$ is a relaxation of $Q$, i.e. $P\supseteq Q$,  we must have $\rho\geq 1$.

Given a graph $G=(V,E)$, let ${\cal U}(G)=\{U\subseteq V(G): U
\textrm{ is an independent set in G}\}$ be the set of all
independent sets in $G$. The fractional chromatic number
$\chi^f(G)$ is defined as the optimal value of the following
linear program.

\begin{eqnarray}
\label{eq:1}\min & \sum_{U\in {\cal U}(G)} y_U& \\
\label{eq:2}      s.t.&\sum_{U\in {\cal U}(G):v\in U} y_U \geq 1& \forall v \in V(G),\\
\label{eq:3}      &y_U\geq 0 &\forall U\in {\cal U}(G) .
\end{eqnarray}

Observe that $\chi^f(G)\leq \chi(G)$ where $\chi(G)$ is the
chromatic number of graph $G$. Moreover, $\chi^f(G)\geq 2$ unless $E(G)=\emptyset$. We assume that $E(G)\neq \emptyset$ for rest of the paper.

\section{Main Result}

We prove the following result relating the integrality gap of the
vertex cover linear programming relaxation and the fractional
chromatic number of $G$.

\begin{theorem}\label{thm:main1}
Given a graph $G=(V,E)$ the integrality gap of vertex cover LP
relaxation $\rho(G)=2-\frac{2}{\chi^f(G)}$.
\end{theorem}

Before we give the proof of Theorem~\ref{thm:main1}, we make a few observations regarding the implications of the theorem. Since $\chi^f(G)=2$
if and only if $G$ is bipartite we obtain that $\rho(G)=1$, i.e., $P(G)$
is integral if and only if $G$ is bipartite. As observed earlier, this implication is also a corollary of the Hungarian method.
If $G$ is planar then $\chi^f(G)\leq \chi(G)\leq 4$ and we obtain that $\rho(G)\leq
\frac32$.

The following corollary shows that it is NP-hard to approximate the integrality gap and thus also ruling out a PTAS for approximating the integrality gap of the natural LP relaxation for the vertex cover problem.

\begin{corollary}
Computing the integrality gap of the vertex cover linear programming relaxation is NP-hard. Moreover, there exists a constant $c>1$ such that it is NP-hard to approximate the integrality gap within a factor of $c$.
\end{corollary}
\begin{proof}
Since the fractional chromatic number and the chromatic number are within a factor of $O(\log n)$ of each other (see for example Theorem 64.13~\cite{Schrijver}),
it is NP-hard to approximate the fractional chromatic number to a factor better than $n^{\delta}$ for some $\delta>0$~\cite{bellare1994improved}.
As a corollary, we also obtain that it is NP-hard to compute the integrality gap of the vertex cover problem exactly.

We now show that it is even NP-hard to approximate the integrality gap within a factor of $c$ for some constant $c>1$. Khot~\cite{khot} implies that there exist constants $L\geq 3, \beta>1$ such that it is NP-hard to distinguish whether the fractional chromatic
number of a graph $G$ is at most $L$ or at least $\beta L$. 
But observe that when the fractional chromatic number of $G$ is at most $L$, then $\rho(G)\leq 2-\frac{2}{L}$ and when the fractional chromatic number of $G$ is more than $\beta L$, then we have $\rho(G)> 2-\frac{2}{\beta L}$. Therefore, it is NP-hard to distinguish
whether $\rho(G)\leq 2-\frac{2}{L}$ or $ \rho(G)> 2-\frac{2}{\beta L}$ and thus ruling out a $c$-approximation where $c=\frac{2-\frac{2}{\beta L}}{2-\frac{2}{L}}=\frac{\beta L-1}{\beta L-\beta }$.
\end{proof}
\\
We now give the proof of the main result.
\\
\begin{proofof}{ Theorem~\ref{thm:main1}}
 We first prove that $\rho(G)\leq 2-\frac{2}{\chi^f(G)}$.

 Let
$x^*$ be any extreme point of $P(G)$. Lemma~\ref{goemanslem} implies  that it is
 enough to show that $(2-\frac{2}{\chi^f(G)})\cdot x^*\in
 IP(G)$.
Nemhauser and Trotter~\cite{nt} imply that
$x^*_v\in\{0,\frac{1}{2},1\}$ for each $v\in V$.
 Let
\begin{eqnarray*}
 V_0&=&\left\{v\in V(G)|x^*_v=0\right\}\\
V_{\frac{1}{2}}&=&\left\{v\in V(G)|x^*_v=\frac{1}{2}\right\}\\
V_1&=&\left\{v\in V(G)|x^*_v=1\right\}
\end{eqnarray*}
 and
 let $H=G\left[V_{\frac{1}{2}}\right]$  be the graph
induced by $V_{\frac12}$.

Let $\chi^f(H)$ be the fractional chromatic number of $H$ and let
$y^*$ denote an optimum solution achieving the optimum value of
$\chi^f(H)$. 

\begin{claim} For each $U\in {\cal U}(H)$, $V_\frac12\setminus U$ is a vertex cover in $H$ and $\left(V_\frac12 \setminus U\right)\cup V_{1}$
is a vertex cover in $G$.\end{claim}
\begin{proof}
Since $U$ is an independent set in $H$, $V(H)\setminus U=V_{\frac12}\setminus U$ is
a vertex cover in $H$.

 Now, consider any edge $e\in E(G)$. If
$e\cap V_1\neq \phi$ then $e$ is covered by $V_{1}$. Else,
if both endpoints of $e$ are in $V_{\frac{1}{2}}$, then $e\in
E(H)$ and it is covered by $V_{\frac12}\setminus U$. Else, it has at least one endpoint
in $V_0$. But, then it must have the other endpoint in $V_1$ as
$x^*_u+x^*_v\geq 1$ where $e=\{u,v\}$.
\end{proof}

For any set $S\subseteq V$, we denote $\one_S\in \{0,1\}^V$ to be the indicator vector of set $S$. Thus,
$\one_{\left(V_{\frac12}\setminus U\right)\cup  V_1} \in IP(G)$ for each $U\in {\cal U}(H)$. Let
$\lambda_U=\frac{y^*_U}{\chi^f(H)}$. Clearly, $\sum_{U\in {\cal
U}(H)} \lambda_U=1$ since $\sum_{U\in {\cal U}(H)} y^*_U=\chi^f(H)$.

\begin{claim}\label{claim2}

$\left(2-\frac{2}{\chi^f(G)}\right)\cdot x^* \geq \sum_{U\in {\cal U}(H)}
\lambda_U\cdot \one_{\left(V_{\frac12}\setminus U\right) \cup  V_1}$.
\end{claim}

\begin{proof}
Consider any component $v$ of the two vectors in the above
inequality. If $v\in V_0$, both the LHS and the RHS are $0$ and the
inequality holds. For any $v\in V_1$, the LHS is
$\left(2-\frac{2}{\chi^f(G)}\right)\geq 1$ while the RHS is at most $\sum_{U\in {\cal
U}(H)} \lambda_U=1$ and the inequality holds.

Now, let $v\in V_{\frac{1}{2}}$. The component in the LHS
corresponding to $v$ is
$$\left(2-\frac{2}{\chi^f(G)}\right)\cdot\frac{1}{2}=1-\frac{1}{\chi^f(G)}.$$
The RHS of the inequality is
\begin{eqnarray*}
\sum_{U\in {\cal U}(H):v\notin U}
\lambda_U &=&1-\sum_{U\in {\cal U}(H):v\in U} \lambda_U \\
&=&1-\sum_{U\in {\cal U}(H):v\in U} \frac{y^*_U}{\chi^f(H)}\\
&\leq&1-\frac{1}{\chi^f(H)}\\
&\leq& 1-\frac{1}{\chi^f(G)}
\end{eqnarray*}

where the first inequality holds as $\sum_{U\in {\cal U}(H):v\in U}
{y^*_U} \geq 1$ and the last inequality uses the fact that
$\chi^f(H)\leq \chi^f(G)$ as $H$ is a induced subgraph of $G$.
\end{proof}
Thus from Lemma~\ref{goemanslem}, we have $\rho(G)\leq
2-\frac{2}{\chi^f(G)}$.

Now, we prove that $\rho(G)\geq
2-\frac{2}{\chi^f(G)}$.
We show this by constructing a cost function for the vertices of
$G$ and showing that any integral solution is at least
$2-\frac{2}{\chi^f(G)}$ times the objective of the linear program. To
construct the cost function we again use the following linear program,
for the fractional chromatic number of graph $G$ defined in \eqref{eq:1}-\eqref{eq:3}.
Its dual is given as follows.

\begin{eqnarray*}
max & \sum_{v\in V(G)} z_v& \\
      s.t.&\sum_{v\in U} z_v \leq 1& \forall U \in {\cal U}(G),\\
      &z_v\geq 0 &\forall v\in V(G).
\end{eqnarray*}

Let $z^*$ be an optimum solution to the dual linear program. Strong duality implies that the
value of $z^*$ is $\chi^f(G)$. Now, consider the cost function
$c:V\rightarrow \R_+$, where $c_v=z^*_v$. Consider $x^*$ as a fractional vertex cover solution defined as $x^*_v=\frac{1}{2}$
for each $v\in V$. Observe that it is a feasible solution in $P(G)$. Hence, we
obtain that \[\min\{cx:x\in P(G)\}\leq \sum_{v\in
V}c_vx^*_v=\sum_{v\in V}z_v^*\cdot\frac{1}{2}=\frac{\chi^f(G)}{2}\].

Let $\one_I$ be an optimal vertex cover solution in $IP(G)$ under the cost
function $c$ and let $U=V\setminus I$.
Clearly, $U$ is an independent set. Hence, \[c\cdot \one_I=\sum_{v\in
I}c_v=\sum_{v\in V}z^*_v-\sum_{v\in U} z^*_v\geq
\chi^f(G)-1\] where the last inequality holds as $\sum_{v\in U}
z^*_v \leq 1$ from the fact that $U$ is an independent set. Hence,
 \[\rho(G)\geq
\frac{\min \{c^T x:x\in IP(G)\}}{\min \{c^T x: x\in P(G)\}}\geq
\frac{\chi^f(G)-1}{\chi^f(G)/2}=2-\frac{2}{\chi^f(G)}\]
 as claimed.
\end{proofof}

\section{Acknowledgement}

 This research is supported in part by National Science Foundation
grant CCF-1717947.

\end{document}